\documentclass[letterpaper, 10 pt, conference]{ieeeconf}  

\IEEEoverridecommandlockouts                              

\usepackage{hyperref}
\usepackage{epsfig} 
\usepackage{times} 
\usepackage{amsmath} 
\usepackage{amssymb}  
\usepackage{comment}
\usepackage{amsmath}
\usepackage{mathtools}
\usepackage{graphics,graphicx,color}
\usepackage{dsfont}

\newtheorem{theorem}{Theorem}

\newtheorem{lemma}[theorem]{Lemma}
\newtheorem{definition}{Definition}[section]
\newtheorem{assumption}{Assumption}
\newtheorem{remark}{Remark}
\newtheorem{example}{Example}

\newcommand{\Sp}[1]{\text{sp}\left(#1\right)}
\newcommand{\s}{\mathcal{S}}
\newcommand{\Trace}{\text{tr}}

\newcommand{\N}{\mathcal{N}}
\newcommand{\TX}{\b{\Phi}_{x}}
\newcommand{\TU}{\b{\Phi}_{u}}
\newcommand{\tx}{\Phi_{x}}
\newcommand{\tu}{\Phi_{u}}

\renewcommand{\b}[1]{\mathbf{#1}}

\newcommand{\I}{\mathcal{I}}
\newcommand{\B}{\mathcal{B}}
\newcommand{\A}{\mathbb{A}}

\newcommand{\X}{\mathcal{X}}

\newcommand{\rTX}{\tilde{\b{\Phi}}_{x}}
\newcommand{\rTU}{\b{\tilde{\Phi}}_{u}}
\newcommand{\bxb}{\tilde{\b{\Phi}}_{x,b}^j}

\newcommand{\inxb}{\tilde{\b{\Phi}}_{x,n}^j}

\newcommand{\bx}{\tilde{\Phi}_{x,b}^j}
\newcommand{\bu}{\tilde{\Phi}_{u,b}^j}
\newcommand{\inx}{\tilde{\Phi}_{x,n}^j}

\newcommand{\0}{\underbar{0}}
\newcommand{\rtu}{\tilde{\Phi}_{u}}
\newcommand{\rtx}{\tilde{\Phi}_{x}}
\newcommand{\AK}{A_K^\ell}
\newcommand{\BK}{B_K^\ell}
\newcommand{\CK}{C_K^\ell}
\newcommand{\DK}{D_K^\ell}

\title{\LARGE \bf
Localized and Distributed $\mathcal{H}_2$ State Feedback Control}

\author{Jing Yu$^{1}$, Yuh-Shyang Wang$^{2}$ and James Anderson$^{3}$
\thanks{$^{1}$ Department of Computing and Mathematical Science, California Institute of Technology, Pasadena, CA. jing@caltech.edu}
\thanks{$^{2}$ Argo AI, Pittsburgh, PA. yswang@argo.ai }
\thanks{$^{3}$ Department of Electrical Engineering and the Data Science Institute at Columbia University, New York, NY. james.anderson@columbia.edu}}

\begin{document}

\maketitle
\thispagestyle{empty}
\pagestyle{empty}

\begin{abstract}
        Distributed linear control design is crucial for large-scale cyber-physical systems. It is generally desirable to both impose information exchange (communication) constraints on the distributed controller, and to limit the propagation of disturbances to a local region without cascading to the global network (localization). Recently proposed System Level Synthesis (SLS) theory provides a framework where such communication and localization requirements can be tractably incorporated in controller design and implementation. In this work, we derive a solution to the localized and distributed $\mathcal{H}_2$ state feedback control problem without resorting to Finite Impulse Response (FIR) approximation. Our proposed synthesis algorithm allows a column-wise decomposition of the resulting convex program, and is therefore scalable to arbitrary large-scale networks. We demonstrate superior cost performance and computation time of the proposed procedure over previous methods via numerical simulation.

\end{abstract}


\section{Introduction}
\label{sec:intro}
Large-scale interconnected systems often demand control designs that comply with structural constraints with respect to communication and interaction. These requirements become especially crucial in engineering applications such as power grids \cite{fang2011smart} and vehicle platoons \cite{li2015overview}. 
Collectively, the challenge of designing controllers subject to these constraints is referred to as distributed or structured control \cite{han2003lmi}. It is known that distributed control problems are in general non-convex. Special cases of distributed control problems, such as those satisfying Quadratic Invariance (QI) \cite{rotkowitz2005characterization}, have been shown to have an exact convex reformulation. 
Therefore, previous works mostly focus on structured controller design in the QI setting. As noted in \cite{wang2014localized}, QI requires global information exchange for strongly connected plants such as a chain system. This imposes limitations on the scalability of the synthesis procedure, and the implementation of the distributed controllers. Particularly, \cite{wang2019system} explored cases where one wishes to go beyond QI conditions and observed that solutions leveraging QI can be more complex to synthesize than its central counterpart \cite{lamperski2015optimal}, thus not scalable to large-scale networks. As the state dimensions of the control systems grow, two control design requirements emerge: (1) \textit{Localization}: It it desirable that the effects of disturbances are limited to a predefined local region without cascading to the global network. (2) \textit{Distributed implementation}: Controller implementation needs to be distributed, allowing only sparse and local information to be exchanged between controllers. 

The first requirement is crucial for systems such as power grids where cascading failures can cause socioeconomic devastation \cite{hines2009cascading}. The second requirement might be imposed even when global information is available to local controllers as computation of local control actions using global information can become intractable in large networks. In this article, we tackle the class of structured control problems subject to these two constraints. In particular, we focus on the state feedback $\mathcal{H}_2$ optimal control setting. Under the QI framework, a large body of work has developed solutions to the state feedback $\mathcal{H}_2$ control problems subject to information sharing constraints \cite{LesL15,fardad2014design,kashyap2019explicit}. However, the \textit{localization} constraints were only recently considered in \cite{anderson2019system} and \cite{anderson2017structured} and motivate further investigation.  




In this work, we present a scalable solution to the \textit{localized} and \textit{distributed} $\mathcal{H}_2$ optimal control problem. We extend previous results that use finite-horizon approximation \cite{wang2014localized,anderson2017structured} to the infinite-horizon case and relieve several assumptions such as the block diagonal control matrix. Further, We provide details of the distributed implementation and computation of the controller leveraging the System Level Synthesis parameterization of the closed-loop maps \cite{wang2019system,ho2020system}. The resulting controller confines disturbances in a local neighborhood while constraining the information exchange among subsystems to a user-specified pattern. 
\paragraph*{Notation}Latin letters $x\in \mathbb{R}^n$ and $A \in \mathbb{R}^{m\times n}$ present vectors and matrices respectively. $A(i,j)$ refers to the $(i,j)^{\text{th}}$ element of the matrix. We use $A(:,j)$ and $A(j,:)$ to refer to the $j^{\text{th}}$ column and $j^{\text{th}}$ row of $A$ respectively. Bold font $\b{x}$ denotes the signal vector sequence $\b{x}:=\{x[t]\}_{t=1}^\infty$. Transfer matrices $\mathbf{G}(z) \in \mathbb{C}^{n\times m}$ have a spectral decomposition $\mathbf{G}(z) = \Sigma_{i = 0}^{\infty} z^{-i}G[i]$ where $G[i] \in \mathbb{R}^{n \times m}$. 
The $j^{\mathrm{th}}$ standard basis vector is $e_j \in \mathbb{R}^n$.
$\Sp{\cdot}$ is the support of a matrix.
For two binary matrices $S_1,S_2 \in \{0,1 \}^{m\times n}$, the operation $S_1 \cup S_2$ performs an element-wise OR operation. Given the matrix $A$, we say $\Sp{A} \subseteq S_1$ if $\Sp{A} \cup S_1 = S_1$. 
We abbreviate the set $\{1,2,\dots,N\}$ as $[N]$ for $N \in \mathbb{N}$.

\section{The Localized and Distributed $\mathcal{H}_2$ Problem}
We consider interconnected systems consisting of $N$ subsystems. For each subsystem $i$, let $x^i \in \mathbb{R}^{n_i}$, $u^i \in \mathbb{R}^{m_i}$, $w^i \in \mathbb{R}^{n_i}$ be the local state, control, and disturbance vectors respectively. Each subsystem $i$ has dynamics:
\begin{equation*}
        \label{eq:local_sys}
        x^i[t] = \sum_{j\in \N^x(i)} A^{ij}x^j[t-1] + \sum_{j \in \N^u(i)} B^{ij}u^j[t-1] + w^i[t]
\end{equation*}
where we write $j \in \N^x(i)$ if the states $x^j$ of subsystem $j$ affect the states of subsystem $i$ through the open-loop network dynamics. Similarly, we denote $j\in \N^u(i)$ if the control action $u^j$ of subsystem $j$ influence the states of subsystem $i$. In addition, the open-loop network interconnection pattern will be denoted as $\A\in \{0,1\}^{N \times N}$: 
\begin{equation*}
\A(i,j) = \begin{cases}
        1 & \text{if } j \in \N^x(i)\\
        0 & \text{otherwise}.
\end{cases}
\end{equation*}
Stacking the dynamics of all subsystems, we can represent the global network dynamics as 
\begin{equation}
        \label{eq:network_sys}
        x[t] = Ax[t-1] + Bu[t-1] + w[t].
\end{equation}
\begin{example}
        \label{ex:1}
        Consider  a chain network as shown in Figure \ref{fig:example1}. Each subsystem $i$ has its local plant $P_i$ and controller $C_i$ with scalar states $x^i$ and control actions $u^i$. For each $i$, $\N^x(i)$ only contains its nearest neighbors.
        \begin{figure}
                \centering
                \includegraphics[scale = 0.43]{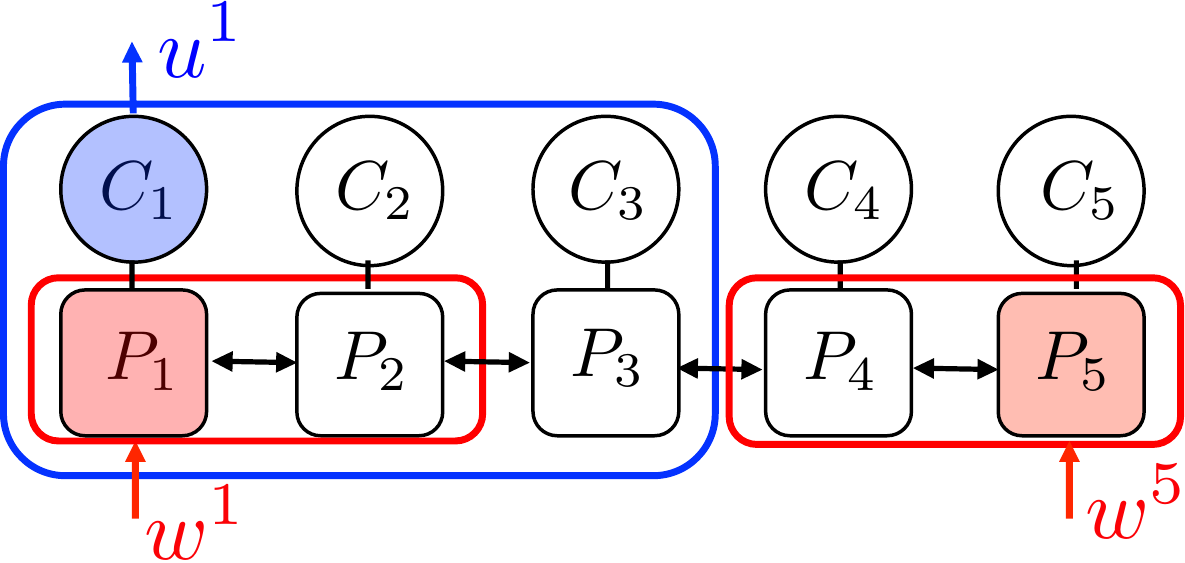}
                \caption{Scalar chain network for dynamic \eqref{eq:chain} with localization and communication requirement that $\s^L = \A$ and $\s^C = \s^{L,e}$. }
                \label{fig:example1}
        \end{figure}
        The stacked network dynamics \eqref{eq:network_sys} for this system has tri-diagonal state propagation matrix $A$ and diagonal $B$ matrix:
        \begin{align}
                \label{eq:chain}
                A &= 
                \begin{bmatrix}
                        *&*&0&0&0\\
                        *&*&*&0&0\\
                        0&*&*&*&0\\
                        0&0&*&*&*\\
                        0&0&0&*&*
                \end{bmatrix} \quad 
                B = \begin{bmatrix}
                        *&0&0&0&0\\
                        0&*&0&0&0\\
                        0&0&*&0&0\\
                        0&0&0&*&0\\
                        0&0&0&0&*
                \end{bmatrix}.
        \end{align}
\end{example}

\subsection{Localization}
\label{sec:localization}
It is often desirable to limit the effects of disturbances in \eqref{eq:network_sys} to a local region for a large network. One may specify the disturbance localization pattern with a binary matrix.
\begin{definition}[Disturbance Localization]
        \label{def:localization}
        The closed-loop of \eqref{eq:network_sys} is said to satisfy disturbance localization according to $\s^L \in \{0,1\}^{N\times N}$ if the following holds: Disturbance $\b{w}^j$ entering subsystem $j$ can propagate to the states $\b{x}^i$ at subsystem $i$ if and only if $\s^L(i,j)\not = 0$.
\end{definition}

\begin{example}
        \label{ex:localization}
        As an example of Definition \ref{def:localization}, consider Figure \ref{fig:example1} with dynamics \eqref{eq:chain}. Let us constrain the closed-loop localization of this chain network to $\s^L = \A$. This means that each local disturbance $\b{w}^i$ can only spread to the set $\N^x(i)$. According to the sparsity of $\A$ in \eqref{eq:chain}, the closed-loop satisfies disturbance localization according to $\s^L$ if disturbance $\b{w}^1$ entering at subsystem $P_1$ can propagate only to $P_2$, while $P_3$ through $P_5$ remain unaffected. Similarly, perturbations entering at $P_5$ only disturb $P_4$ and $P_5$.
\end{example}

We call the subsystems that can be affected by $\b{w}^i$ the \textit{localized region} of $\b{w}^i$. Elements in the localized region of $\b{w}^i$ corresponds to the nonzero elements of the $i^{\text{th}}$ \textit{column} of $\s^L$.
In example \ref{ex:1}, localized region for  $\b{w}^i$ is $\N^x(i)$. An equivalent requirement of disturbance localization per Definition \ref{def:localization} is that the ``boundary" subsystems of each localized region remain at zero to prevent disturbances from propagating outside of the localized region. To this end, we formalize the notion of the boundary subsystems.
\begin{definition}[Extended Localization Pattern]
        \label{def:extend}
Given sparsity pattern $\s^L$ for disturbance localization, the extended localization pattern is $\s^{L,e} = \Sp{\A\s^L}$.
\end{definition}

Matrix $\s^{L,e}$ can be interpreted as the propagation of $\s^L$ according to dynamics \eqref{eq:network_sys} if no action were to be taken to contain the spread of disturbances. 
We now define the boundary subsystems for a given localization pattern $\s^L$.
\begin{definition}[Boundary Subsystems]
        \label{def:boundary}
        The set of the boundary subsystems for the localized region of $\b{w}^i$ is 
        \begin{equation*}
                \B(i) := \{j\in [N]\,\,|\,\, \s^{L,e}(j,i) - \s^{L}(j,i) \not = 0\}.
        \end{equation*}
\end{definition}
Intuitively, the set $\B(i)$ for the localized region of $\b{w}^i$ contains the indices of the bordering subsystems that controls the spread of the disturbance from within the localized region to the outside of the region. 
\begin{example}
\label{ex:boundary}
We continue with Example \ref{ex:localization} with dynamics \eqref{eq:chain} where $\s^L = \A$. With Definition \ref{def:extend} and \ref{def:boundary}, we have:
\begin{equation*}
    \s^{L,e} = \begin{bmatrix}
    1&1&1&0&0\\
    1&1&1&1&0\\
    1&1&1&1&1\\
    0&1&1&1&1\\
    0&0&1&1&1
    \end{bmatrix}, 
    \s^{L,e} - \s^{L} = \begin{bmatrix}
    0&0&1&0&0\\
    0&0&0&1&0\\
    1&0&0&0&1\\
    0&1&0&0&0\\
    0&0&1&0&0
    \end{bmatrix}.
\end{equation*}
The boundary index set $\B(i)$ thus corresponds to the position of nonzero elements on the $i^{\text{th}}$ \textit{column} of $\s^{L,e} - \s^{L}$. For instance, $\B(3) = \{1,5\}$ and $\B(1) = \{3\}$.
\end{example}


\subsection{Distributed Implementation}
\label{sec:communication}
Controllers for large networks are generally required to have distributed implementation. This means that each \textit{local} controller for subsystems only has access to information from its neighboring subsystems. We denote the information about subsystem $j$ at time $t$ as $\I^j_t$ that includes all past states, control actions and controller internal states at subsystem $j$ up to time $t$. Given \textit{a priori} specified sparsity pattern for communication among subsystems, we have: 
\begin{definition}[Distributed Communication]
        \label{def:comm}
        A controller $\b{K}$ for \eqref{eq:network_sys} is said to conform to the communication constraint $\s^C \in \{0,1\}^{N\times N}$ if the following holds: Subsystem $i$ at time $t$ has access to information set $\I^j_t$ from subsystem $j$ for all $t \in \mathbb{N}$ if and only if $\s^C(i,j) \not = 0$.
\end{definition}


\begin{example}
        Consider Figure \ref{fig:example1} with dynamics \eqref{eq:chain}. Let the communication pattern in this case be $\s^C = \s^{L,e}$ as shown in Example \ref{ex:boundary}, which is the minimum communication requirement for $\s^L$ to be achievable. At every time step, control actions $\b{u}^1$ generated by subsystem $1$ depends on the information from subsystem 1,2,and 3 as shown in Figure \ref{fig:example1}. Similarly, we need information from subsystem 1, 2, 3, and 4 for $\b{u}^2$ by the sparsity pattern of $\s^{L,e}$. 
\end{example}

\subsection{Problem Statement}
\label{sec:problem}
We now state the localized and distributed state feedback $\mathcal{H}_2$ problem. We want to minimize the $\mathcal{H}_2$ performance index of output $\b{z} = Q^{\frac12}\b{x} + R^{\frac12} \b{u} $ of the closed-loop of \eqref{eq:network_sys}. The disturbance, i.e., the $w[t]$'s are assumed to be independently and identically distributed and drawn from $\N(0,I)$, and $Q^{\frac12},R^{\frac12} \succ 0$. Denote $\b{x} = \{x[t]\}_{t=0}^\infty$, $\b{u} = \{u[t]\}_{t=0}^\infty$. The objective is to search for a controller that localizes the closed-loop response and possesses a distributed implementation. We write this as the following optimization problem:
\begin{subequations}
        \begin{align}
                & \underset{\b{K}}{\text{minimize}} 
                & & \mathbb{E}_{w[t] \sim \mathcal{N}(0,I)}\left\|\begin{bmatrix}
                        Q^{\frac12} & 0\\ 0 & R^{\frac12}
                \end{bmatrix}  
                \begin{bmatrix}
                        \b{x} \\\b{u}
                \end{bmatrix} 
                \right\|_{\mathcal{L}_2}^2  \label{eq:h2} \tag{P0} \\
                & \text{subject to} 
                & & x[t] = Ax[t-1] + Bu[t-1] + w[t]\nonumber\\
                &&& \b{u} = \b{K}\b{x},\,\, \b{K} \text{ internally stabilizing} \nonumber\\
                &&&  \b{K} \text{ localizes closed-loop according to $\s^L$} \label{eq:localization-const} \\
                &&& \b{K} \text{ conforms to the communication constraint}\nonumber \\
                &&& \,\,\,\,\,\, \text{ according to $\s^C$.} \label{eq:comm-const} 
        \end{align}
\end{subequations}
where $\|\b{x}\|^2_{\mathcal{L}_2} := \sum_{k = 0}^\infty \|x[k]\|_2^2$ denotes the norm on signals in the $\mathcal{L}_2$ space. We assume $(A,B)$ is stabilizable.
Problem \ref{eq:h2} has practical application in large-scale cyber-physical systems such as power systems \cite{wang2018separable,tseng2020deployment}.
We note that in contrast to all previously formulated SLS problems, there is no FIR constraint in~\ref{eq:h2}.

\section{Preliminaries on System Level Synthesis}
\label{sec:sls}
Before the development of the solution to Problem \eqref{eq:h2}, we first review the System Level Synthesis framework \cite{anderson2019system} that has seen much success in distributed \cite{wang2018separable}, nonlinear \cite{yu2020achieving}, MPC \cite{alonso2020explicit}, and adaptive \cite{ho2019scalable} control design.

Consider the closed-loop dynamics of \eqref{eq:network_sys} under a linear feedback law $\b{u} = \b{K}\b{x}$. We denote the closed-loop mappings (CLMs) from disturbance $\b{w}$ to $\b{x}$ and $\b{u}$ by $\TX$,$\TU$ respectively, i.e., $\left[\begin{array}{c}\mathbf{x}\\\mathbf{u}\end{array}\right] = \left[\begin{array}{c}\TX\\\TU\end{array}\right] \mathbf{w}$. Let $N_x = \sum_i^N n_i$ and $N_u = \sum_i^N m_i$. Then $\TX(k,j)$ and $\TU(l,j)$ are the impulse response transfer function from $w(j)$ to $x(k)$ and $u(l)$ for $k \in \left[N_x\right]$ and $l \in \left[N_u\right]$. The closed-loop mappings $\TX$ and $\TU$ can be explicitly represented as $\TX = z(zI-A-B\b{K})^{-1}$ and $\TU = z\b{K}(zI-A-B\b{K})^{-1}$ after performing the $\mathcal{Z}$-transformation of the closed loop of \eqref{eq:network_sys}. Note we have followed convention in nonlinear SLS theory \cite{ho2020system} where $\TX$ and $\TU$ are \textit{causal} operators.
The System Level Synthesis (SLS) framework introduces a novel parametrization of all such achievable CLMs under internally stabilizing controllers $\b{K} \in \mathcal{RH}_\infty$. Crucially, SLS allows re-parameterization of any stabilizing controllers to be expressed and implemented with CLMs. Instead of searching for controller $\b{K}$, one looks for desirable closed-loop responses $\TX,\TU$ and recovers the controller transfer function that realizes these closed-loop behaviors as $\b{K} = \TU (\TX)^{-1}$. 
This is formalized as the following result adapted from \cite{wang2019system}. 
\begin{theorem}[\cite{wang2019system}]
        \label{thrm:SLS}
        For the dynamics \eqref{eq:network_sys}, the affine subspace in variables $\TX$ and $\TU$ defined by 
\begin{subequations}
\label{eq:feasibility}
        \begin{align}
                \tx[0] &= I, \quad \TX,\TU \in \mathcal{RH}_\infty  \label{eq:td-feasibility1}\\
                \tx[t+1] &= A\tx[t] + B\tu[t], \label{eq:td-feasibility2}
        \end{align}
\end{subequations}
             characterizes all closed-loop mappings achievable by an internally stabilizing controller. Moreover, for any $\TU,\TX$ satisfying \eqref{eq:feasibility}, controller $\mathbf{K} = \TU (\TX)^{-1}$ achieves the desired closed-loop responses $\TX$,$\TU$, is internally stabilizing, and can be implemented equivalently as 
            \begin{subequations}
        \label{eq:td-controller}
        \begin{align}
        u[t] &=  \sum_{k = 0}^t \tu[k] \hat{w}[t-k], \quad \hat{w}[0] = x[0] \label{eq:sls-1}\\
        \hat{w}[t+1] &= x[t+1] - \sum_{k = 1}^{t+1} \tx[k] \hat{w}[t+1-k]  \label{eq:sls-2}
\end{align}
\end{subequations}
            where $\b{\hat{w}}$ is the internal state of the controller.
\end{theorem}

Controller \eqref{eq:td-controller} can be regarded as estimating past disturbances in \eqref{eq:sls-2} and acting upon the estimated disturbances according to a specified closed-loop mapping $\TU$ in \eqref{eq:sls-1}. An important consequence of Theorem \ref{thrm:SLS} is that any structures imposed on the closed-loop responses $\TX,\TU$ satisfying \eqref{eq:feasibility}, such as sparsity constraints on the spectral elements of $\TX,\TU$, trivially translate into structures on the realizing controllers \eqref{eq:td-controller} that achieve the designed responses. 

Note constraint \eqref{eq:localization-const} and \eqref{eq:comm-const} in \ref{eq:h2} can be equivalently expressed in terms of the CLMs of the closed loop of \eqref{eq:network_sys}. We first define what it means for CLMs of \eqref{eq:network_sys} to conform to localization and communication sparsity patterns. 
\begin{definition}[Sparsity of CLMs]
        \label{def:clm-sparsity}
        A CLM $\TX \in \mathbb{C}^{\sum_i n_i \times \sum_i n_i}$ for \eqref{eq:network_sys} satisfies $\TX \in \s^L$ if for all $k \in \mathbb{N}$, $\Sp{\tx[k]}$ is a block matrix with $(i,j)^{\text{th}}$ block being $\mathds{1}_{n_i \times n_j}$ when $\s^L(i,j) = 1$, and $\0_{n_i \times n_j}$ when $\s^L(i,j) = 0$. 
        Similarly, $\TU \in \s^C$ for $\TU \in \mathbb{C}^{\sum_i m_i \times \sum_i n_i}$, if for all $k$, $\Sp{\tu[k]}$ is a block matrix with $(i,j)^{\text{th}}$ block being $\mathds{1}_{m_i \times n_j}$ when $\s^C(i,j) = 1$, and $\0_{m_i \times n_j}$ when $\s^C(i,j) = 0$.
\end{definition}

Constraint \eqref{eq:localization-const} is equivalent to requiring $\TX \in \s^L$ by definition. On the other hand, note that controller \eqref{eq:td-controller} inherits the communication pattern from the sparsity of $\TX,\TU$. Therefore, constraint \eqref{eq:comm-const} can be expressed as $\TU \in \s^C$. 



\section{Main Results}
\label{sec:main}
We derive the solution to Problem \eqref{eq:h2} in two parts. First, we present the \textit{synthesis} of the localized and distributed controller via CLMs using the SLS parameterization. The synthesis procedure naturally decomposes into smaller problems, allowing computation to only involve local information, thus favorably scales to large networks. The second part of the solution investigates the \textit{implementation} of the localized and distributed controller. We make explicit how decomposed local controllers subject to communication constraints achieve the global objective of stabilization and localization. 

\subsection{\textbf{Synthesis} of CLMs}
\label{sec:synthesis}
\subsection*{Step 1: Re-parameterization with CLMs} 
We substitute variables $\TU \b{w}$ and $\TX \b{w}$ in place of $\b{x}$ and $\b{u}$ as the optimization variable in Problem \eqref{eq:h2} by definition of CLMs. An equivalent re-parameterization is as follows:
\begin{subequations}
\begin{align}
& \underset{\TX,\TU}{\text{minimize}} 
& & \sum_{k = 0}^{\infty} \Trace\left( \tx[k]^TQ\tx[k]+\tu[k]^TR\tu[k] \right) \label{eq:step1} \tag{P1} \\
& \text{subject to}
& & \eqref{eq:td-feasibility1} , \eqref{eq:td-feasibility2}, \quad \TX \in \s^L, \quad \TU \in \s^C \nonumber
\end{align}
\end{subequations}
where 
we simplify the objective function in \eqref{eq:h2} using the fact that i.i.d. white noise $\b{w}$ has identity covariance. As addressed in Section \ref{sec:sls}, \eqref{eq:td-feasibility1} and \eqref{eq:td-feasibility2} characterize the space of CLMs achievable by an stabilizing controller $\b{K}$, thus replacing the equality constraints in \eqref{eq:h2}. 

\subsection*{Step 2: Column-wise Decomposition}
As a feature of SLS problems, \eqref{eq:step1} can be decomposed in a column-wise fashion when $\s^L$ and $\s^C$ are appropriately chosen \cite{wang2018separable}. The columns of $\TX$ and $\TU$ can be solved for in parallel and reconstructed to recover the solution to \eqref{eq:step1}. 
For each column $j\in [N_x]$, we denote $\TX^j$ and $\TU^j$ as the $j^{\text{th}}$ column of $\TX$ and $\TU$. The decomposed problem \eqref{eq:step1} for each $j\in[N_x]$ has the form:
\begin{subequations}
       \begin{align}
       & \underset{\TX^j,\TU^j \in \mathcal{RH}_\infty}{\text{minimize}} 
       & & \sum_{k = 0}^{\infty}  \tx^j[k]^TQ\tx^j[k]+\tu^j[k]^TR\tu^j[k] \label{eq:step2} \tag{P2} \\
       & \text{subject to} 
       & & \tx^j[0] = e_j \label{eq:initial}\\
       &&&\tx^j[t+1] = A \tx^j[t] + B \tu^j[t] \label{eq:col-sys} \\
       &&& \TX^j \in \s^L(:,j), \quad \TU^j \in \s^C(:,j).\label{eq:col-sparsity} 
       \end{align}
\end{subequations}
Recall that the $(k,j)^{\text{th}}$ position of $\TX$ represents the closed-loop transfer function from $\b{w}(j)$ to $\b{x}(k)$ with $k\in[N_x]$. Within the column vector $\TX^j$, we can identify $\TX^j(k)$ with position $k$'s associating to the states in subsystems in $\B(j)$. Moreover, since column vector $\TX^j$ and $\TU^j$ are constrained to the $j^{\text{th}}$ column of prescribed sparsity patterns $\s^L$ and $\s^C$ respectively, we can reduce \eqref{eq:step2} by removing zero entries other than those associated with $\B(j)$. We denote the reduced column vectors that contains the entries associated with $\B(j)$ as $\rTX^j$ and $\rTU^j$. Similarly, the problem parameters $A$, $B$, $Q$, $R$ can be reduced by selecting submatrices $A^{(j)}$, $B^{(j)}$, $Q^{(j)}$, and $R^{(j)}$ consisting of columns and rows associated with the boundary entries and non-zero entries of $\TX^j$ and $\TU^j$. Note these sub-matrices now contain only dynamics information from subsystems that are allowed to transmit information to the  $j^{\text{th}}$ state's subsystem. We further rearrange the reduced vectors and matrices in \eqref{eq:col-sys} by grouping the entries associated with boundary subsystems as follows:
\begin{equation}
\label{eqn:bdy}
\underbrace{\begin{bmatrix}
\inx\\ \bx
\end{bmatrix}[k+1]}_{\rtx^j[k+1]}
 = 
\underbrace{\begin{bmatrix}
A^{(j)}_{nn} & A^{(j)}_{nb}\\
A^{(j)}_{bn}  & A^{(j)}_{bb} 
\end{bmatrix}}_{A^{(j)}}
\begin{bmatrix}
        \inx\\ \bx
\end{bmatrix}[k]  + 
\underbrace{\begin{bmatrix}
B^{(j)}_{n} \\
B^{(j)}_{b} 
\end{bmatrix}}_{B^{(j)}}
\rtu^j[k]
\end{equation}
where $\bxb$ denotes the entries on column vector $\rTX^j$ that are associated with $\B(j)$ and $\inxb$ represents the nonzero entries of $\rTX^j$ that are not associated with boundary subsystems. Here, $A^{(j)}$ and $B^{(j)}$ are partitioned accordingly. With abuse of notation, We overload $\rTU^j $ to denote the rearranged and reduced vector $\TU^j$.
\begin{example}
        \label{ex:reduction}
        Consider the scalar chain example in Figure \ref{fig:example1} for the local problem with $j = 4$, \textit{i.e.}, the subproblem \eqref{eq:step2} corresponding to the fourth column of $\TX,\TU$. We have the constraint $\TX^4 = [0,$ $0,$ $\TX(3,4),$ $\TX(4,4),$ $\TX(5,4)]^T$ according to the fourth column of localization pattern $\s^L = \A$.  In this case, we have $\b{\tilde{\Phi}}^4_{x,b}= [\TX(2,4)]^T$ defined in Definition \ref{def:boundary} and $\b{\tilde{\Phi}}^4_{x,n} = [\TX(3,4), \,\, \TX(4,4), \,\, \TX(5,4)]^T$. Therefore, the rearranged and reduced vector is $\b{\rtx}^4 = [\TX(3,4),$ $\TX(4,4)$, $\TX(5,4)$, $\TX(2,4)]^T$.
\end{example}

Note that the first part of constraint \eqref{eq:col-sparsity} now becomes equivalent to the requirement that $\bxb$ remains at origin at all time for the localized region of $\b{w}^j$. This is because of the ``initial condition" \eqref{eq:initial}. By keeping the entries associated with boundary subsystems at zero, we implicitly impose that for all $k$, $ \Sp{A\tx^j[k]  +  B \tu^j[k]} \subseteq \s^L(:,j)$, which is necessary and sufficient to ensure $\TX^j \in \s^L(:,j)$. Therefore, the local problem \eqref{eq:step2} after rearrangement becomes
\begin{subequations}
        \begin{alignat}{2}
        \underset{\rTX^j,\rTU^j \in \mathcal{RH}_\infty}{\text{min}} \,\,\, \sum_{k = 0}^{\infty} & \rtx^j[k]^TQ^{(j)}\rtx^j[k]+\rtu^j[k]^TR^{(j)}\rtu^j[k]  \label{eq:step3} \tag{P3}\\
         \text{subject to}   \,\,\,\,\,\,\,\,\,\,\,\,\,\,\,\,\,\,\,\,& \rtx^j[0] = e_{j_j} ,\quad \eqref{eqn:bdy} \label{eq:init}\\
        & \bx[k] = 0, \forall k \label{eq:boundary}
        \end{alignat}
\end{subequations}
where $j_i$ denotes the new position of element $\TX(j,i)$ in the rearranged and reduced vector $\tilde{\b{\Phi}}_x^i$. Vectors $e_{j_i}$ have the same dimension as $\rTX^i$. We differentiate the position of element $\TX(j,i)$ in $\b{\tilde{\Phi}}^i_{x,n}$ with the notation $\tilde{j}_i$. Vectors $e_{\tilde{j}_i}$ has the same dimension as $\b{\tilde{\Phi}}^i_{x,n}$.
\begin{example} 
       Continuing Example \ref{ex:reduction} where $i,j = 4$, then $\TX(4,4)$ is in the \textit{second} position in rearranged and reduced vector $\b{\rtx}^4$. Thus, $j_4 = 2$, $e_{j_4} = \left[0, \, 1, \, 0, \, 0\right]^T$, and $\tilde{j}_4 = 2$ with $e_{\tilde{j}_4} = \left[0, \, 1, \, 0\right]^T$. Consider instead $j = 4$ and $i = 5$, then $\TX(4,5)$ is in the \textit{first} position in $\b{\tilde{\Phi}}^5_{x} = [\TX(4,5),\, \TX(5,5),\, \TX(3,5)]^T$ while it is also in the \textit{first} position in $\b{\tilde{\Phi}}^5_{x,n} = [\TX(4,5),\, \TX(5,5)]^T$. We then have $j_5 = 1$ with $e_{j_5} = [1,\,0,\,0]^T$ and $\tilde{j}_5 = 1$ with $e_{\tilde{j}_5} = [1,\,0]^T$.
\end{example}

\subsection*{Step 3: De-constraining Subproblems}
We now de-constrain \eqref{eq:step3} by characterizing CLMs that satisfy \eqref{eq:boundary}. We first substitute \eqref{eq:boundary} into \eqref{eqn:bdy} in \eqref{eq:step3} and conclude that \eqref{eq:boundary} is equivalent to requiring 
\begin{equation}
        \label{eq:lin-eq}
        -B^{(j)}_b \rTU^j = A^{(j)}_{bn} \inxb .   
\end{equation}
Due to the equality constraint \eqref{eq:init} and \eqref{eqn:bdy}, the free optimization variable is $\rTU^j$ in \eqref{eq:step3}. Therefore, \eqref{eq:lin-eq} has solutions $\rTU^j$ if and only if the following assumption holds:
\begin{assumption}
        \label{existence}
        $B^{(j)}_bB^{(j)\dagger}_b = I$.
\end{assumption}
Recall that constraint \eqref{eq:boundary} is sufficient and necessary for the CLMs to comply to the localization pattern $\s^L$. This means assumption \ref{existence} is the minimum requirement for the each local problems \eqref{eq:step3} to be localizable according to the local neighborhood specified by $\s^L$. Further, per Definition \ref{def:boundary}, the number of boundary subsystems can generally be less than the total dimension of control actions, \textit{i.e.}, $B_b^{(j)}$ is a wide matrix.

\begin{lemma}
        \label{lemma1}
        Under Assumption \ref{existence}, the parametrization 
        \begin{equation}\label{eqn:M-param}
        \rtu^j[k] = -B^{(j)\dagger}_b A^{(j)}_{bn} \inx[k] + \left( I-B^{(j)\dagger}_b B^{(j)}_b \right)v^{j}[k]\end{equation} 
        with $v^{j}[k]$ a free vector variable characterizes all $\rtu^j[k]$ that satisfies \eqref{eq:boundary}. 
\end{lemma}
\begin{proof}
 Under Assumption \ref{existence}, \eqref{eq:lin-eq} has solutions of the form \eqref{eqn:M-param}. This can be checked by confirming that $\text{Range}\left(I-B^{(j)\dagger}_b B^{(j)}_b\right)= \text{Kernel}\left(B_b^{(j)}\right)$. Substituting \eqref{eqn:M-param} in \eqref{eqn:bdy}, one can verify that $\bx[k] = 0$, $\forall k=1,2,\dots$. 
\end{proof}

The re-parametrization of optimization variable $\bu$ in \eqref{eq:step3} with $v^{j}$ allows us to express an equivalent local optimization problem without \eqref{eq:boundary}. Substitute \eqref{eqn:M-param} into \eqref{eq:step3}:
\begin{subequations}
\begin{align}
& \underset{\inxb,\b{v}^{j} \in \mathcal{RH}_\infty }{\text{min}}
&  \sum_{k = 0}^{\infty}& \inx[k]^T\tilde{Q}^{(j)}\inx[k]+v^{j\,T}[k]\tilde{R}^{(j)}v^{j}[k] \nonumber \\
& \text{subject to} 
& & \inx[0] = e_{\tilde{j}_j} \label{eq:step4} \tag{P4}\\
&&& \inx[k+1] = \tilde{A}^{(j)}\inx[k] + \tilde{B}^{(j)}v^{j}[k]  \nonumber
\end{align}
\end{subequations}
where
\begin{align*}
  \tilde{R}^{(j)} &= \left(\left(R^{(j)}\right)^{\frac12}\left(I-B^{(j)\dagger}_b B^{(j)}_b\right)\right)^T \\
  &\,\,\,\,\,\,\,\,\,\,\,\left(\left(R^{(j)}\right)^{\frac12}\left(I-B^{(j)\dagger}_b B^{(j)}_b\right)\right)\\
        \tilde{Q}^{(j)} &= \left((Q^{(j)})^{\frac12} -(R^{(j)})^{\frac12} B^{(j)\dagger}_b A_{bn}^{(j)}\right)^T \\
        &\,\,\,\,\,\,\,\,\,\,\,\left((Q^{(j)})^{\frac12}-(R^{(j)})^{\frac12}B^{(j)\dagger}_b A_{bn}^{(j)}\right) \\
    \tilde{A}^{(j)} &= A^{(j)}_{nn}-B^{(j)}_n B^{(j)\dagger}_b A^{(j)}_{bn}\\
    \tilde{B}^{(j)} &= B^{(j)}_n \left(I-B^{(j)\dagger}_bB^{(j)}_b \right).
\end{align*}

\subsection*{Step 4: Local Riccati Solutions}
For each column $j$ with $j \in [N_x]$, problem \eqref{eq:step4} can be treated as an infinite horizon LQR problem with which an optimal "control policy" $\tilde{K}^{(j)}$ can be computed in closed form via discrete-time algebraic Riccati equation (DARE): 
\begin{align*}
    \tilde{K}^{(j)} = -\left(\tilde{R}^{(j)} + \tilde{B}^{(j)T} X^{(j)}\tilde{B}^{(j)}\right)^{-1}\tilde{B}^{(j)T} X^{(j)}\tilde{A}^{(j)},
\end{align*}
where $X^{(j)}$ is the Riccati solution to the DARE:
\begin{align*}
    X^{(j)} =  &\,\, \tilde{Q}^{(j)} + \tilde{A}^{(j)T}X^{(j)}\tilde{A}^{(j)}- \tilde{A}^{(j)T}X^{(j)}\tilde{B}^{(j)} \nonumber \\
    & \left(\tilde{R}^{(j)}+\tilde{B}^{(j)T}X^{(j)}\tilde{B}^{(j)}\right)^{-1}\tilde{B}^{(j)T}X^{(j)}\tilde{A}^{(j)}.
\end{align*}
With optimal solutions $v^j[k] = \tilde{K}^{(j)} \inx[k]$ to \eqref{eq:step4}, solutions to \eqref{eq:step3} can be recovered via \eqref{eqn:M-param} as:
\begin{align}
        \inx[0] &=\,\, e_{\tilde{j}_j}  \label{eq:recover}\\
        \rtu^j[k] &= \left(-B^{(j)\dagger}_b A^{(j)}_{bn} + \left(I-B^{(j)\dagger}_b B^{(j)}_b\right)\tilde{K}^{(j)}\right)\inx[k] \nonumber \\
        \inx[k] &= \left(\tilde{A}^{(j)} + \tilde{B}^{(j)}\tilde{K}^{(j)}\right)\inx[k-1]. \nonumber
\end{align}
Note the optimal solution to \eqref{eq:step4} via the Riccati equation is stable, so $\b{v}^j$ and $\inxb$ construct stable and proper transfer matrices.

In summary, we went through a series of transformations and decompositions from the original localized and distributed state feedback $\mathcal{H}_2$ problem \eqref{eq:h2} to \eqref{eq:step4}. Indeed, given solutions to the local problems \eqref{eq:step4}, solutions to \eqref{eq:h2} can be recovered. In particular, we define embedding operator $E_x(\cdot)$ and $E_u(\cdot)$ that apply padding of zero's to the reduced vectors $\inxb$ and $\rTU^j$ by assigning entries of $\inxb$ and $\rTU^j$ to the positions of nonzero elements of $\tx(:,j)$ and $\tu(:,j)$ such that $E_x\left( \inxb \right) \in \mathbb{R}^{N_x}$ and $E_u\left( \rTU^j \right) \in \mathbb{R}^{N_u}$. 
\begin{example}
        \label{ex:padding}
        Consider the reduced vector $ \tilde{\b{\Phi}}^4_{x,n} = \left[\TX(3,4), \,\, \TX(4,4), \,\, \TX(5,4)\right]^T$ for $j = 4$ in Example \ref{ex:reduction}. Applying the embedding operator, we have that $E_x\left(\tilde{\b{\Phi}}^4_{x,n}\right) = \left[ 0, \,0 \,,\TX(3,4), \,\, \TX(4,4), \,\, \TX(5,4)  \right]^T$, which recovers $\TX^4$ respecting the sparsity of $\s^L(:,4)$. Similarly, $e_{\tilde{j}_j} = \left[0, \, 1, \, 0\right]^T$ and $E_x\left(e_{\tilde{j}_j}\right) = e_j = \left[ 0, \, 0,\, 0,\, 1,\, 0 \right]^T$.
\end{example}
\begin{theorem}
        Let $\b{\Phi}_x^*$ be the column-wise concatenation of $E_x\left(\inxb \right)$ and let $\b{\Phi}_u^*$ be the column-wise concatenation of  $E_u \left(\rTU^j \right)$ with $\inxb$'s and $\rTU^j$'s recovered from the solution to \eqref{eq:step4} via \eqref{eq:recover}. Then  $\b{\Phi}_x^*$ and  $\b{\Phi}_u^*$ minimize \eqref{eq:step1}.
\end{theorem}
\begin{proof}
        It is straight forward to check that optimization \eqref{eq:step1} is an instance of \textit{column-wise separable problem} (Section III, \cite{wang2018separable}) where both the objective function and constraints are column-wise separable and can be partitioned and solved in columns as in \eqref{eq:step2} in parallel. Therefore, solutions to subproblem \eqref{eq:step2} can be concatenated to recover the solution to \eqref{eq:step1}. Note that by construction, $E_x\left(\inxb \right) = \TX^j$ and  $E_u \left(\rTU^j \right) = \TU^j$ comprise the optimal solution to \eqref{eq:step2} for each $j$. Concatenate $E_x\left(\inxb \right)$'s and $E_u \left(\rTU^j \right)$'s in a column-wise fashion and the resulting matrices are solutions to \eqref{eq:step1}.
\end{proof}

\subsection{Controller Realization \& Implementation}
A second design requirement is the distributed implementation of the the controller that achieves localized closed-loop. Given CLMs $\TX$, $\TU$ synthesized in Section \ref{sec:synthesis}, we can directly conclude that \textit{theoretically}, $\b{K} = \TU\left(\TX\right)^{-1}$ with implementation \eqref{eq:td-controller} achieves the given CLMs $\TX,\TU$ and conforms to the communication constraint according to $\s^C$. This is because the inheritance of sparsity structures of the controller implementation from CLMs by Theorem \ref{thrm:SLS}. Interested readers are referred to \cite{tseng2019deployment} for in-depth discussion on implementation of SLS controllers for cyber-physical systems. However, due to the state-space form of solutions from \eqref{eq:step4}, \textit{practical} implementation of a controller that achieves the \textit{theoretical} global CLMs remains elusive. 

We decompose the global SLS controller \eqref{eq:td-controller} into $N_x$ sub-controllers using the solution to \eqref{eq:step3}. The global control action $u[t]$ can be accordingly decomposed into $ N_x$ "sub-control actions". These sub-control actions will be computed using solutions from \eqref{eq:step3}. These sub-control actions are then assembled together to form a global control action. Importantly, the computation of each sub-control action conforms to communication constraint $\s^c$. We now make precise of this high-level description.

To ease notation, we denote $x_\ell[t] \in \mathbb{R}$ and $w_\ell[t] \in \mathbb{R}$, for $\ell \in \left[N_x\right]$ as the $\ell^{\text{th}}$ position in the state and disturbance vector $x[t]$ and $w[t]$ in the global dynamics \eqref{eq:network_sys}, respectively. Further, we define the indices associated with the state vector $x^j \in \mathbb{R}^{n_j}$ of subsystem $j \in [N]$ as $\X(j) := \{ \ell \in \left[N_x\right] \,|\, x_\ell \in x^j \}$. Thus, $\X(j)$'s partition the global state vector $x[t]$ in \eqref{eq:network_sys} into $N$ sets containing the states associated with the $N$ subsystems. Conversely, we use $\X^{-1}(\ell)$ to denote the subsystem index to which state $x_\ell$ belongs.

For each $\ell \in \left[N_x\right]$, we compute the sub-control action vector $u_\ell$, which has the same vector dimension as $\tilde{\Phi}_u^\ell$, as:
\begin{subequations}
        \label{eq:sub-controller}
        \begin{align}
                \hat{w}_\ell[t] &= x_\ell[t] - \sum_{i \in \N^w(\ell)  } \xi_i[t] \left(\tilde{\ell}_i\right) \label{eq:subcontroller1}\\
                \xi_\ell[t+1] &= \AK \xi_\ell[t] + \BK \hat{w}_\ell[t] \label{eq:k-1}\\
                u_\ell[t] &= \CK \xi_\ell[t] + \DK \hat{w}_\ell[t], \label{eq:k-2}
        \end{align}
\end{subequations}
where $\hat{w}_\ell[t] \in \mathbb{R}$ can be considered as an estimate of disturbance $w_\ell[t]$. Internal state $\xi_\ell[t]$ of each sub-controller has the same dimension as $\tilde{\Phi}^\ell_{x,n}[t] $ and $\xi_i[t] \left(\tilde{\ell}_i\right)$ denotes the $\tilde{\ell}_i^{\,\,\text{th}}$ element in the internal state vectors $\xi_i$. Note that controller internal variables have initial condition $\hat{w}_\ell[0] = x_\ell[0]$ and $\xi_\ell[0] = \underbar{0}$. We also define the set $\N^w(\ell)$ as $\N^w(\ell):= \left\{ i \in \left[ N_x\right] |\, S^L\left(\X^{-1}(\ell), \X^{-1}(i)\right) \not = 0 \right\}.$
In particular, the set $\N^w(\ell)$ contains global indices $i \in \left[ N_x \right] $ such that $x_i$ is a state that is allowed to communicate its information to the subsystem that contains state $x_\ell$, conforming to communication pattern $\s^C$. The compliance to the communication constraint is due to the fact that $\s^L \subset \s^C$.

\begin{figure}[h]
        \centering
        \includegraphics[scale = 0.65]{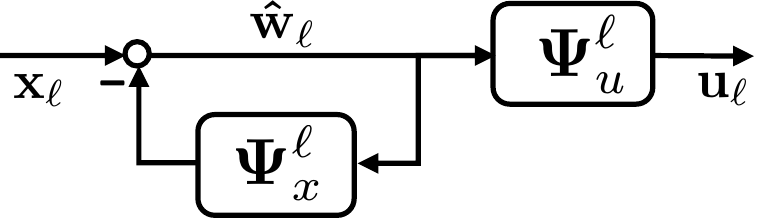}
        \caption{Column-wise sub-controller implementation for global controller $\b{K} = \TU(\TX)^{-1}$. $\b{x}_\ell$ is the $\ell^{\text{th}}$ state, $\hat{\b{w}}_\ell$ is the estimated $\ell^{\text{th}}$ disturbance, and $\b{u}_\ell$ is the sub-control actions induced by $\ell^{\text{th}}$ state's deviation from origin.}
        \label{fig:controller}
\end{figure}
Equation \eqref{eq:k-1} and \eqref{eq:k-2} are the sub-controller internal dynamics specified by $\left(\AK,\BK,\CK,\DK\right)$ that takes in estimated disturbance $\hat{w}_\ell$ and output decomposed control actions $u_\ell$. The internal dynamics for the $l^{\mathrm{th}}$ sub-controller are:
\begin{align*}
\AK &= \tilde{A}^{(\ell)} + \tilde{B}^{(\ell)}\tilde{K}^{(\ell)},\quad \BK = \left( \tilde{A}^{(\ell)} + \tilde{B}^{(\ell)}\tilde{K}^{(\ell)} \right) e_{\tilde{\ell}_\ell}\\
\CK &= -B^{(\ell)\dagger}_b A^{(\ell)}_{bn} + \left(I-B^{(\ell)\dagger}_b B^{(\ell)}_b\right)\tilde{K}^{(\ell)}\\
\DK &= \left( -B^{(\ell)\dagger}_b A^{(\ell)}_{bn} + \left(I-B^{(\ell)\dagger}_b B^{(\ell)}_b\right)\tilde{K}^{(\ell)} \right) e_{\tilde{\ell}_\ell}.
\end{align*}
Referring to \eqref{eq:recover}, it is straight forward to verify that \eqref{eq:sub-controller} is indeed the state space realization of each decomposed SLS controller implementing the reduced $\ell^{\text{th}}$ column of $\TX$ and $\TU$ synthesized from \eqref{eq:step3}. In particular, \eqref{eq:sub-controller} implements a transfer function mapping from scalar signal $\b{x}_\ell$ to vector signal $\b{u}_\ell$. Further, each sub-controller is stable since $\AK$ is Hurwitz. The block diagram of this transfer function is shown in Figure \ref{fig:controller}, where:
        \begin{equation}
                \label{eqn:local-controller}
        \b{\Psi}^\ell_x = \left[
        \begin{array}{c|c}
        \AK & \BK \\
        \hline
         I & 0 \\ 
        \end{array}
        \right],
        \quad
        \b{\Psi}^\ell_u = \left[
        \begin{array}{c|c}
        \AK & \BK  \\
        \hline
        \CK&
        \DK \\ 
        \end{array}
        \right].
        \end{equation}
For each state $\ell^{\text{th}}$ state $\b{x}_\ell$ deviating from the origin due to disturbance $\b{w}_\ell$, it invokes subsystems $j \in \N^C(\ell)$ to transmit information among each other in order to generate a \textit{collaborative} sub-control action $\b{u}_\ell$ from these subsystems. Moreover, internal dynamics \eqref{eq:k-1}, \eqref{eq:k-2} of each $\ell$ sub-controller involves only the global dynamics associated with subsystems $j \in \N^C(\ell)$. Therefore, by definition of $\N^C(\ell)$, we conclude that each sub-controller's implementation conforms to the communication pattern specified by $\s^C$. By the superposition property of the input-output behaviors of linear systems, we can sum over all the sub-control actions induced by each $\b{w}_\ell$ and the global control action $u[t] \in \mathbb{R}^{ N_u}$ is: 
\begin{equation}
        \label{eq:global-controller}
        u[t] = \sum_{i = 1}^{N_x} E_u(u_\ell[t]),
\end{equation}
where each sub-control action $\b{u}_\ell$, which has the same vector dimension as $\tilde{\Phi}_u^\ell$ can be appropriately padded with zeros using the linear operator $E_u(\cdot)$ to recover a vector dimension in $\mathbb{R}^{ N_u}$ as in Example \ref{ex:padding}.

The following result confirms that collectively, the sub-controllers indeed achieve the prescribed global behaviors.
\begin{theorem}
        controller implemented \eqref{eq:sub-controller} and \eqref{eq:global-controller} defined by solutions to \eqref{eq:step3} is internally stabilizing for \eqref{eq:network_sys} and achieves the closed-loop mappings $\TX$ and $\TU$ constructed by stacking in a column-wise fashion the solutions to \eqref{eq:step3}.
\end{theorem}
\begin{proof}
Recall Theorem \ref{thrm:SLS}, where an internally stabilizing controller that realizes given closed-loop maps $\TX$ and $\TU$ has centralized implementation \eqref{eq:td-controller}. Therefore, we establish the equivalence between global control action $u[t]$ generated from \eqref{eq:td-controller} and $u[t]$ generated from \eqref{eq:global-controller}. Consider \eqref{eq:sls-2} where the controller's internal state $\b{\hat{w}}$ has dynamics
\begin{align*}
        \hat{w}[t] &= x[t] - \sum_{k = 1}^{t} \tx[k] \hat{w}[t-k] \\
                  &= x[t] - \sum_{i = 1}^{N_x} \sum_{k = 1}^{t} \tx^i[k] \hat{w}_i[t-k].
\end{align*}
For each $\ell^{\text{th}}$ position in $\hat{w}[t]$, due to the localization sparsity pattern $\s^L$ imposed on $\TX$, the scalar dynamics is 
\begin{align*}
        \hat{w}_\ell[t] = x_\ell[t] - \sum_{i \in \N^w(\ell)} \sum_{k = 1}^{t} \tx(\ell,i)[k] \,\hat{w}_i[t-k].
\end{align*}
Since $\TX^i$ for all $i \in [N_x]$ are recovered from \eqref{eq:recover} via the linear operators $E_x(\cdot)$, it is straight forward to verify that 
$$\sum_{k = 1}^{t} \tx(\ell,i)[k] \,\hat{w}_i[t-k] = \xi_\ell[t](\tilde{\ell}_i), \quad \text{for } t = 1,2,\dots.$$
We therefore conclude that \eqref{eq:sls-2} and \eqref{eq:subcontroller1},\eqref{eq:k-1} are equivalent. Similarly, re-write \eqref{eq:sls-1} as
$$ u[t] =  \sum_{i=1}^{N_x} \sum_{k = 0}^t \tu^i[k] \hat{w}_i[t-k] .$$
According to \eqref{eq:recover}, one can check that
$\sum_{k = 0}^t \tu^i[k] \hat{w}_i[t-k] = E_u(u_\ell[t]),$
thus verifying the equivalence between \eqref{eq:sls-1} and \eqref{eq:k-1},\eqref{eq:k-2},\eqref{eq:global-controller}. 
\end{proof}

The intuition behind sub-controllers is that at every time step, the global controller actions are decomposed into $\ell^{\text{th}}$ sub-control actions that only attenuate the $\ell^{\text{th}}$ disturbance, \textit{i.e.}, $\b{w}_\ell$. Therefore, whenever $\b{w}_\ell$ enters the system, only subsystems in the localized region of this disturbance reacts, computing the sub-control actions using only local information available according to $\s^L$.

\begin{remark}
        We presented solution to the localized and distributed $\mathcal{H}_2$ problem under instantaneous information exchange among subsystems according to $\s^C$. Our methodology can be extended to the case where information transmission is delayed. In particular, one can employ the state-space augmentation by introducing fictitious relay subsystems that have trivial dynamics and do not have associated cost nor noise \cite{lamperski2015optimal}. An efficient representation and computation of solution to delayed systems will be future work.
\end{remark}

\begin{remark}
        Results in this paper concurrently solve the general qudratic-cost infinite-horizon SLS problem with explicit state space controller implementation \cite{wang2019system}, while previous results have \cite{jensen2020topics} considered specific problems such as consensus. 
        
\end{remark}

\section{Simulations}
\label{sec:simulation}
We validate our results and highlight the advantage of the proposed infinite-horizon $\mathcal{H}_2$ controller\footnote{ The code for the simulation can be found at \href{https://github.com/jy-cds/infinite_horizon_sls.git}{this GitHub repository}}. Consider a bi-directional scalar chain system parametrized by $\alpha$ and $\rho$: 
\begin{equation*}
        x^i[t+1] = \rho (1-2\alpha ) x^i[t] + \rho \alpha \sum_{j \in \{ i\pm 1\}} x^j[t] + u^i[t] + w^i[t]
    \end{equation*}
The parameter $\rho$ characterizes the stability of the overall system while $\alpha$ decides how coupled the dynamics between each node is. The $i^{\text{th}}$ state in the global state vector $\b{x}$ is dynamically coupled to its nearest neighbors.The localization and communication constraints in this case are chosen to be $(A,d)$-sparse and $(A,d+1)$-sparse, respectively (for details, see section II-B in \cite{wang2014localized}) with $d$ specifying how many neighbors a disturbance can spread to. 

Figure \ref{fig:cost_comparison} shows that that the proposed infinite-horizon $\mathcal{H}_2$ controller outperforms previous FIR SLS controllers, which uses finite-horizon approximation to solve for suboptimal controllers to the localized and distributed $\mathcal{H}_2$ problem. As the finite horizon grows larger, the FIR SLS controller's cost approaches the optimal cost achieved by the proposed method. 

Figure \ref{fig:cost_vs_dimension} demonstrates the computation time reduction from the proposed method, compared to previous finite-horizon solutions. Since both finite-horizon SLS controller \cite{wang2018separable} and the proposed method can be synthesized in a distributed and localized way, we compare the computation time where all columns of the CLMs solutions are computed in \textit{parallel}. In general, each of the parallel $N_x$ subproblem for FIR SLS controller computation involves $N_x(N_x + N_u)T$ optimization variables where $T$ is the FIR horizon. On the other hand, the proposed infinite-horizon method only requires $N_x$ parallel solutions to the Riccati equations of size $\tilde{n}_i$'s, which are the sizes of the reduced columns in \eqref{eq:step4}.

\begin{figure}[h]
 \centering
 \includegraphics[clip, trim=1.6cm 8.3cm 0.5cm 9cm,scale = 0.43]{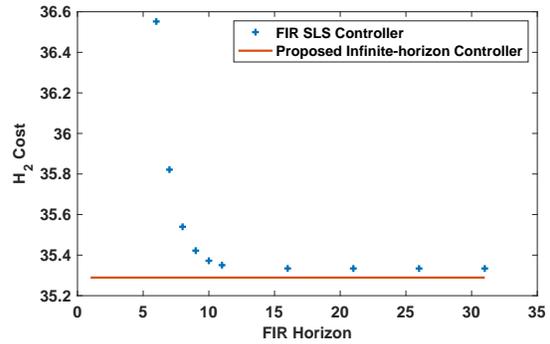}
 \caption{$\mathcal{H}_2$ Cost comparison between the FIR SLS controller \cite{wang2019system} and the infinite-horizon SLS controller proposed in this paper. Clearly, as the FIR horizon becomes larger, the FIR controller's cost converges to the infinite-horizon optimal controller. In this example, we have an 20-node unstable chain system with $\alpha = 0.4$ and $\rho = 1.25$ with 50\% actuation where only every other subsystem has control authority ($\b{u}^i \not = 0$). We impose $(A,d)$ and $(A,d+1)$ sparsity on the localization and communication pattern respectively with $d = 5$. Note when FIR horizon is less than $T = 6$, it is infeasible to find the FIR SLS controller. }
 \label{fig:cost_comparison}
\end{figure}

\begin{figure}[h]
 \centering
 \includegraphics[clip, trim=1.5cm 8.3cm 0.5cm 9cm,scale = 0.43]{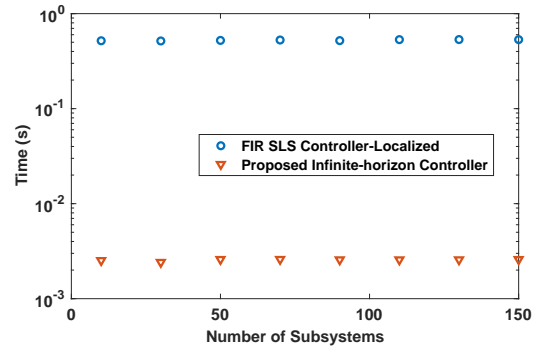}
 \caption{The proposed infinite-horizon SLS controller has a clear advantage in computation time over FIR SLS controllers. Here we have fixed $T = 10$ as the FIR horizon.}
 \label{fig:cost_vs_dimension}
\end{figure}

\section{Conclusion}
We propose and derive the solution to the localized and distributed $\mathcal{H}_2$ problem in this paper. Our result generalizes previous methods that uses finite-horizon approximation and make explicit the distributed implementation of the controller. In particular, the derivation in this paper also present an infinite-horizon SLS controller.

\section{ACKNOWLEDGMENTS}
J.Y. thanks Dimitar Ho for helpful discussions.

\bibliographystyle{IEEEtran}
\bibliography{reference.bib}

\end{document}